\documentclass{article}
\usepackage[tbtags]{amsmath}
\usepackage{amsfonts,verbatim,amssymb}
\usepackage{mathrsfs}
\usepackage{amsthm}
\allowdisplaybreaks

\theoremstyle{plain}
\newtheorem{theorem}{Theorem}

\theoremstyle{definition}

\begin{document}

%%\udk{517.5}

\date{}

\title{\Large\bf The Consistency of Quantum Mechanics\\
Implies Its Non-Determinism}

\author{\bf Iegor Reznikoff\\[1ex]
 Professor Emeritus, Departement de Philosophie,\\
Universit\'e de Paris-Ouest, 92001 Nanterre, France\\
\textit{E-mail}: dominiqueleconte@yahoo.fr}

%\address{Professor Emeritus, Departement de Philosophie,
%Universit\'e de Paris-Ouest, 92001 Nanterre, France}

%\email{dominiqueleconte@yahoo.fr}

\maketitle

%%\begin{abstract}
%%This is an abstract.
%%\end{abstract}

%%\begin{keywords}  %%undefined
%%\end{keywords}
%%%%%%%%%%%%%%%%%%%%%%%%%%%%%%%%%%%%%%%%%%
%%% ----------------------------------------------------------------------

%%\section*{\large Introduction}
% ------------------------------------------------------------------------

Recently [1], Conway and Kochen have shown that, in Quantum Mechanics(QM),
there is a mathematical correlation between the \textit{freedom} of the observer to observe (or not) a particle and the indetermination of some property of the particle in the theory; this is essentially the meaning 
of their Free-Will Theorem (FWT). The introduction of the notion of observer's freedom has shed a new light on the subject. Conway and Kochen used four axioms (three from QM and one of geometrical nature), and their proof involves some subtle, mostly physical, arguments. In [2] we have given a purely logical proof of the FWT, based on only two of their axioms (one from QM and the geometrical one, but expressed in formal logic). Here we show that the hypothesis of observer's freedom is not necessary for the conclusion, provided the consistency
(non-contradiction) of~QM is assumed. This paper is self-contained, the main ideas, however, come from [2].

\section*{\large Introduction: the Logical Setting}

Let $F$ be the physico-mathematical theory in which QM are expressed; we consider $F$ as a deductive theory in classical logic with $\vdash$ as the corresponding deductive symbol. $A(t, i)$ is a predicate symbol for `at time~$t$, $A$ observes a particle and verifies (measures) some property (value) depending on the parameter (or direction) $i$\,'; and $P(t, i)$ is 
for `at time~$t$ the particle has the observed property on the same parameter~$i$\,'.
We suppose that $F$ contains the following two axioms:
\medskip

$\text{Ob}$\hfill$\forall t\ \forall i\ \ (A(t,i)\to P(t,i))$\hfill $\hphantom{A}$
\medskip

$\text{Im}$\hfill$\forall t\ \urcorner\,\forall i\ \ P(t,i)$,
\qquad where $i = 1, 2,\dots,n$.\hfill$\hphantom{A}$
%%
%%where $i = 1, 2,\dots,n$.
\medskip

\noindent
Ob (for \textit{observational}) and~Im (for \textit{impossibility}) are, respectively, Axioms~2 and~3 in~[2], where their meaning in QM is explained. Briefly, Ob is given by the experimental (observational) evidence in QM (e.g. by observation and measurement of the spin values in different directions $i$ in an EPR-like experiment or by separate measurements of position ($i$ = 1) or momentum ($i$ = 2) of a particle), while Im is given by mathematical elementary probability theory (Bell's theorem~[3]) or geometrical evidence (Kochen--Specker theorem, as it is used in the proof of Conway and Kochen) or given by physical evidence (Heisenberg uncertainty principle: it is impossible to observe position and momentum at the same time).

The notion of \textit{freedom} of the observer, relative to the deduction in theory~$F$, is defined naturally (as in [2]) by
\medskip

$\hphantom{A}$\hfill$\forall t\ \forall i\ \ (F\nvdash A(t,i)\wedge  F\nvdash\,\urcorner\, A(t,i))$.\hfill$\hphantom{A}$
\medskip

\noindent
while the indetermination or \textit{freedom} (following Conway and Kochen) of the particle relative to the observed property in the theory~$F$ is defined by
\medskip

$\hphantom{A}$\hfill$\forall t\ \exists i\ \ (F\nvdash P(t,i)\wedge  F\nvdash\,\urcorner\, P(t,i))$.\hfill$\hphantom{A}$
\medskip

\noindent
Following Conway and Kochen, we say briefly that the (behaviour of the) \textit{particle is free } in~$F$, for the more correct statement that the observed property of the particle is undetermined in~$F$.

In [2] we have proved the main non-deterministic theorem:

%Theorem 1
\begin{theorem}
\label{th1}
If a consistent theory~$F$ verifies axioms \textup{Ob} and~\textup{Im},
then the freedom of the observer implies the freedom of the particle relative to theory~$F$.
\end{theorem}

This applies to the theory~$F$ of~QM, where the freedom of the observer (resp. particle) is verified if $F = F(t')$ contains only events that happen at time $t' < t$  (where $t$ is the time of observation): the behaviour of the observer is not predetermined by the theory  (and, consequently, says the theorem, nor is the behaviour of the particle for some parameter). For 
$n = 2$, in the case of the Heisenberg uncertainty principle, the theorem shows that the freedom or non-determinism of the behaviour of particles could have been established on logical grounds since the very beginning of~QM.

The proof of this theorem is not needed in what follows.
\bigskip

\section*{\large Consistency and Non-Determinism}

It is important to note, first, that the physical theory~$F$ containing QM, as a whole, contains the mathematics needed for real (and complex) analysis necessary to express physics and contains no more axioms on infinity than those of Analysis. Consequently, the consistency of~$F$, as a whole, supposes the consistency of the mathematics used. While, concerning non-determinism, the incompleteness of mathematics (e.g. real analysis) is well known, we are concerned here only with the non-determinism of physical events. We write $\operatorname{Cons}$ = $\operatorname{Cons}$ $F$ for the formula expressing the consistency of~$F$ as is usual in Mathematical Logic (formally, it means that $F\nvdash  0 = 1$). The purpose here is to eliminate the hypothesis of observer's freedom; for this Ob is modified in the following way:
\medskip

$\text{COb}$\hfill$\forall t\ \forall i\ \ ((\operatorname{Cons}\to A(t,i))\to P(t,i))$.\hfill $\hphantom{A}$
\medskip

The axiom COb is obviously stronger than Ob; its meaning has to be explained. If $F$ is consistent, the two axioms are equivalent. Literally, in COb the observation is conditional: if, whenever $F$ is consistent, $A$ observes the particle,
then the particle verifies property $P(t, i)$. But $\operatorname{Cons}$ $F$ does not depend on time; thus, believing in the consistency of~$F$, informally, COb has the same meaning as Ob. Moreover it is easy to prove formally that
$\operatorname{Cons}\vdash\operatorname{COb}\leftrightarrow \operatorname{Ob}$
and  $\vdash\operatorname{Cons} F\leftrightarrow\operatorname{Cons}(F + \operatorname{COb})$ (recall that $F$ contains Ob and~Im). If $\operatorname{Cons}$ is true
(as will be assumed) COb reduces to Ob, and if $\operatorname{Cons}$ is false, COb reduces to  $\forall\,t\ \forall\,i\ \ P(t, i)$, which contradicts Im;
consequently, in this case, $F$, as expected, is false (contradictory). Moreover, the notion of freedom vanishes if the theory is not consistent. 

Believing in $\operatorname{Cons}$, there are no reasons in physics not to accept COb if we accept Ob, which belongs to experimental physics (in the cases mentioned above: EPR, Heisenberg's, or others). But, formally, provably, we have

%Theorem 2
\begin{theorem}
\label{th2}
If $F$ is consistent, then \textup{COb} implies that there are free particles relative to the theory~$F$.
\end{theorem}

%Proof.
\begin{proof}
Suppose the particle is not free, then (using the definition of freedom of the particle)
\medskip

$(*)$\hfill$\exists t\ \forall i\ \ (F \vdash P(t,i)
\vee F\vdash\,\urcorner\, P(t,i))$.\hfill $\hphantom{A}$
\medskip

But $F\vdash P(t, i)$ for all $i$ contradicts Im; therefore, for some $i$,
\medskip

$\hphantom{A}$\hfill$\exists t\ \ (F\vdash\,\urcorner\,P(t,i))$.\hfill$\hphantom{A}$
\medskip

\noindent
Now, because of COb, this implies
\medskip

$\hphantom{A}$\hfill$\exists t\ \exists i\ \ (F\vdash\,\urcorner\,(\operatorname{Cons}\to A(t,i))$;\hfill$\hphantom{A}$
\medskip

\noindent
hence, for some $t$ and $i$, $F\vdash\operatorname{Cons}\wedge\,\urcorner\, A(t,i)$
and $F\vdash\operatorname{Cons}F$. However, because of G\"odel's theorem,
this is impossible if $F$ is consistent.
\end{proof}

\section*{\large Commentary and Improvement}

Thus, if we assume informally (that's only what we can do) that the QM theory~$F$ is consistent and we accept COb, which in this case, informally, reduces to Ob (included in QM), then formally the determinism~($*$)  of~$F$ implies the inconsistency of~$F$. This contradiction shows that~($*$) is impossible and that there are free behaviors of particles in~QM.

Now, the question naturally arises: Which real logical or physical part in~$F$ implies contradiction? If we suppose $F$ \textit{provably} consistent, then, of course,
$\operatorname{Cons}$ is false and implies any formula. For example, for all~$F$, trivially,    $\urcorner\operatorname{Cons}F$, $F\vdash\perp$ (always false); hence
$F\vdash\,\urcorner\,\urcorner\,\operatorname{Cons}$ and, classically, $F\vdash\operatorname{Cons}$. But this is based on the assumption that, provably, $\operatorname{Cons}$ is true, which is, of course, not the case in the proof of the theorem. However, we have here, formally and \textit{not using the meaning of} $\urcorner\,\operatorname{Cons}$ or $\urcorner\operatorname{Cons}F$, $F\vdash\perp$:
\medskip

$(**)$\hfill$\urcorner\operatorname{Cons}, \operatorname{COb},
\operatorname{Im}\vdash\perp$,\hfill $\hphantom{A}$
\medskip

\noindent
because $\urcorner\operatorname{Cons}$, together with
\medskip

$\hphantom{A}$\hfill$\forall t\ \forall i\ \ ((\operatorname{Cons}\to A(t,i))\to P(t,i))$,\hfill$\hphantom{A}$
\medskip

\noindent
imply $\forall t\ \forall i\ \ P(t,i)$, which contradicts Im. And, from  (**),
therefore we obtain   $\operatorname{COb},
\operatorname{Im}\vdash\, \operatorname{Cons}$.
It is then possible to argue that the conclusion $F\vdash\operatorname{Cons}F$ in the proof results simply from COb and~Im and not genuinely from QM (in~$F$), although, in the proof, $F$ appears from the \textit{physical} assumption~($*$) and the contradiction follows
from $F\vdash\,\urcorner\, P(t, i)$.

To avoid this difficulty, axiom Im can be superseded by the following axiom:
\medskip

$\text{CIm}$\hfill$\operatorname{Cons}\to\ \text{Im}\quad   
\text{or, explicitly,} \quad 
\operatorname{Cons} \rightarrow  \forall t\ \urcorner\,\forall i\ \ P(t,i)$
\hfill $\hphantom{A}$

\bigskip
This axiom is unquestionable, because it is a consequence of Im,
and~Im is, as for Kochen--Specker's (or resp. Bell's) theorem, an \textit{elementary} geometrical---especially, in the simple proof using Peres construction
(see [1])---(or resp. elementary probabilistic) result; and if the theory is consistent, then necessarily, it requires Im, since      $\urcorner\,$Im implies $\urcorner\,\operatorname{Cons}$.

But, with CIm, we have no equivalent of~ (**)  (without using   $\urcorner\,\operatorname{Cons}\vdash\perp$), and we cannot formally deduce
$\operatorname{Cons}$ from COb and CIm. Nevertheless, if $F_0$ is $F$
without Ob and~Im (thus, 
$F = F_0 + \operatorname{Ob} + \operatorname{Im}$)
and if we set
\medskip

$\hphantom{A}$\hfill$F' = F_0 + \operatorname{COb} + \operatorname{CIm}$,
\hfill$\hphantom{A}$
\medskip

\noindent
then we have
\medskip

%\noindent
$(***)$\hfill$\operatorname{Cons} F\,\vdash\,\operatorname{Cons} F'$\hfill $\hphantom{A}$
\medskip

\noindent
(recall that $\operatorname{Cons}, \operatorname{Ob}\vdash\operatorname{COb}$
and that  $\operatorname{Cons} = \operatorname{Cons} F$)
\medskip

and the following:

%Theorem 3
\begin{theorem}
\label{th3}
If $F$ is consistent, then F' is consistent and there are free particles relative to the theory~$F'$.
\end{theorem}

%Proof.
\begin{proof}
It is essentially the same as for Theorem 2. Suppose we have~($*$)
with $F$ =\nobreak $F'$; then, since $F'\vdash P(t, i)$    for all $i$ and CIm or
$\operatorname{Cons}\to\forall\,t\,\urcorner\, \forall\,i\ \ P(t, i)$
would imply $F'\vdash\,\urcorner\,\operatorname{Cons}$,
while, because of  (***),  this is impossible
(since $\operatorname{Cons}\vdash\,\urcorner\,\operatorname{Cons}F'$ would follow),
we must have, as above,  $F'\vdash\,\urcorner\, P(t, i)$ for some $t$ and $i$,
which, with COb, implies $F'\vdash\operatorname{Cons}F$, and hence
$F'\vdash\operatorname{Cons}F'$ with the same conclusion that this is impossible,
since $F$ and (therefore) $F'$ are assumed consistent.
\end{proof}

It should be remarked that
the hypothesis on the observer's freedom is not necessary in the proof;
apart from the informal assumption of the consistency of~QM, no other hypothesis is needed. The result can be formulated in the following way.

%Theorem 4
\begin{theorem}
\label{th4}
Assuming the theory of Quantum Mechanics to be consistent, there are necessarily physical events that are undetermined in the theory. The incompleteness of Mathematics relative to consistency implies the non-determinism of Physics relative to physical events.
\end{theorem}

This shows how deeply mathematical truth and physical truth are interlinked.

\section*{\large \bf References}
%\begin{\thebibliography}

[1] J.~Conway and S.~Kochen, The Strong Free-Will Theorem, February 2009,
AMS, Providence, RI, vol. 56/2, pp. 226--232.

\noindent
[2] I.~Reznikoff, A Logical Proof of the Free-Will Theorem, \texttt{arXiv:1008.3661v1 [quant-ph] 21 Aug 2010 (http://arxiv.org/abs/1008.3661)}.

\noindent
[3] J.~S.~Bell, On the Einstein--Podolsky--Rosen Paradox,
Physics {\bf 1}, 195--200 (1964).

%\end{\thebibliography}

\end{document}